\newcommand{\longversion}[1]{#1}
\newcommand{\shortversion}[1]{}

\documentclass[sigconf]{acmart}
\AtBeginDocument{%
  }

\setcopyright{acmlicensed}
\copyrightyear{2018}
\acmYear{2018}
\acmDOI{XXXXXXX.XXXXXXX}
\acmConference[Conference acronym 'XX]{Make sure to enter the correct
  conference title from your rights confirmation email}{June 03--05,
  2018}{Woodstock, NY}
\acmISBN{978-1-4503-XXXX-X/2018/06}

\usepackage[utf8]{inputenc} 
\DeclareUnicodeCharacter{2099}{$_n$}
\DeclareUnicodeCharacter{2098}{$_m$}
\DeclareUnicodeCharacter{22EE}{\vdots}
\DeclareUnicodeCharacter{2081}{$_1$}
\DeclareUnicodeCharacter{2082}{$_2$}
\DeclareUnicodeCharacter{2083}{$_3$}
\usepackage{adjustbox} 
\usepackage{todonotes}
\usepackage{xspace}
\usepackage{url}
\usepackage{hyperref}
\usepackage{cleveref}
\usepackage{newtxtext}

\newcommand{\DTASh}{\textsc{\textbf{DTAS}}\xspace}

\usepackage{array}
\usepackage{subcaption}
\usepackage[ruled,vlined,linesnumbered]{algorithm2e}

\SetKwInOut{KwIn}{Input}
\SetKwInOut{KwOut}{Output}

\makeatletter
\g@addto@macro\normalsize{%
  \setlength\abovedisplayskip{4pt}%
  \setlength\belowdisplayskip{4pt}%
  \setlength\abovedisplayshortskip{2pt}%
  \setlength\belowdisplayshortskip{2pt}%
}
\makeatother

\setcopyright{none}
\renewcommand\footnotetextcopyrightpermission[1]{}

\makeatletter
\newenvironment{proofsketch}{%
  \par
  \pushQED{\qed}%
  \normalfont \topsep6\p@\@plus6\p@\relax
  \trivlist
  \item[\@proofindent\hskip\labelsep
        {\scshape Proof Sketch\@addpunct{.}}]\ignorespaces
}{%
  \popQED\endtrivlist\@endpefalse
}
\makeatother

\begin{document}

\title{Fair Distribution of Digital Payments: Balancing Transaction Flows for Regulatory Compliance}

% \author{Ashlesha Hota}
%  \authornote{Both authors contributed equally to this research.}
%  \email{ashleshahota.23@kgpian.iitkgp.ac.in}
% % \orcid{1234-5678-9012}
% \affiliation{%
%    \institution{IIT Kharagpur, India}}
%    \city{Kharagpur}
%    \state{West Bengal}
%    \country{India}
% }
% \author{Shashwat Kumar}
%  \email{kshashwat.iit@gmail.com}
%  \affiliation{%
%    \institution{IIT Kharagpur, India}}
%    \city{Kharagpur}
%    \state{West Bengal}
%    \country{India}
% }

% \author{Daman Deep Singh}
%  \email{damandeepddsb@gmail.com}
% \affiliation{%
%   \institution{IIT Delhi, India}}
  % \city{Delhi}

% \author{Abolfazl Asudeh}
% \email{asudeh@uic.edu}
% \affiliation{%
%   \institution{University of Illinois, United States}}
%   \city{Chicago}
%   \country{United States}
% }

% \author{Palash Dey}
% \email{palash.dey@cse.iitkgp.ac.in}
% \affiliation{%
%  \institution{IIT Kharagpur, India}}
 % \city{Kharagpur}
 % \state{West Bengal}

% \author{Abhijnan Chakraborty}
% \email{abhijnan@cse.iitkgp.ac.in}
% \affiliation{%
% \institution{IIT Kharagpur, India}}
 % \city{Kharagpur}
 % \state{West Bengal}

 \author{Ashlesha Hota}
 \authornote{Both authors contributed equally to this research.}
 %\email{ashleshahota.23@kgpian.iitkgp.ac.in}
\affiliation{%
   \institution{IIT Kharagpur}
%    \city{Kharagpur}
%    \state{West Bengal}
    \country{India}
 }
\author{Shashwat Kumar}
\authornotemark[1]
%\email{kshashwat.iit@gmail.com}
 \affiliation{%
   \institution{IIT Kharagpur}
%    \city{Kharagpur}
%    \state{West Bengal}
   \country{India}
}

\author{Daman Deep Singh}
%\email{damandeepddsb@gmail.com}
\affiliation{%
  \institution{IIT Delhi}
  % \city{Delhi}
  \country{India}}

\author{Abolfazl Asudeh}
%\email{asudeh@uic.edu}
\affiliation{%
  \institution{University of Illinois Chicago}
%   \city{Chicago}
  \country{United States}
 }

\author{Palash Dey}
%\email{palash.dey@cse.iitkgp.ac.in}
\affiliation{%
 \institution{IIT Kharagpur}
 % \city{Kharagpur}
 % \state{West Bengal}
 \country{India}}

\author{Abhijnan Chakraborty}
%\email{abhijnan@cse.iitkgp.ac.in}
\affiliation{%
\institution{IIT Kharagpur}
 % \city{Kharagpur}
 % \state{West Bengal}
 \country{India}}

% \author{Ashlesha Hota}
% \email{ashleshahota.23@kgpian.iitkgp.ac.in}

% \author{Shashwat Kumar}
% \email{kshashwat.iit@gmail.com}

% \author{Palash Dey}
% \email{palash.dey@cse.iitkgp.ac.in}

% \author{Abhijnan Chakraborty}
% \email{abhijnan@cse.iitkgp.ac.in}

% \affiliation{%
%   \institution{Indian Institute of Technology Kharagpur}
%   \country{India}
% }

% \author{Daman Deep Singh}
% \email{damandeepddsb@gmail.com}
% \affiliation{%
%   \institution{Indian Institute of Technology Delhi}
%   \country{India}
% }

% \author{Abolfazl Asudeh}
% \email{asudeh@uic.edu}
% \affiliation{%
%   \institution{University of Illinois at Chicago}
%   \country{United States}
% }

\newcommand{\eqdef}{\overset{\mathrm{def}}{=\joinrel=}}
 \renewenvironment{proof}{\noindent{\sc Proof:}}{ \hfill $\square$\\ }
\allowdisplaybreaks

\newcommand{\pl}{{\bf pl}}

%Color and Pictures
%\usepackage[usenames,svgnames,dvipsnames]{xcolor}

%\usepackage{charter,eulervm}

%\usepackage{charter}
\sloppy

\usetikzlibrary{arrows,positioning}

\newcommand{\defproblem}[3]{
  \vspace{1mm}
\begin{center}
\noindent\fbox{

  \begin{minipage}{\textwidth}
  \begin{tabular*}{\textwidth}{@{\extracolsep{\fill}}l} \textsc{\underline{#1}} \\ \end{tabular*}\vspace{1ex}
  {\bf{Input:}} #2  \\
  {\bf{Question:}} #3
  \end{minipage}
 
  }
\end{center}
  \vspace{1mm}
}

\newcommand{\defparproblem}[4]{
  \vspace{1mm}
\begin{center}
\noindent\fbox{

  \begin{minipage}{\textwidth}
  \begin{tabular*}{\textwidth}{@{\extracolsep{\fill}}lr} \textsc{#1}  & {\bf{Parameter:}} #3 \\ \end{tabular*}
  {\bf{Input:}} #2  \\
  {\bf{Question:}} #4
  \end{minipage}
 
  }
\end{center}
  \vspace{1mm}
}

\newdimen\prevdp
\def\leftlabel#1{\noalign{\prevdp=\prevdepth
   \kern-\prevdp\nointerlineskip\vbox to0pt{\vss\hbox{\ensuremath{#1}}}\kern\prevdp}}

%A useful package if you write url addresses:
%\usepackage{url}

\renewcommand{\labelitemi}{$\vartriangleright$}

\newcommand{\SB}{\mathbb{S}}

\newcommand{\PPAD}{\ensuremath{\mathsf{PPAD}}\xspace}
\newcommand{\NP}{\ensuremath{\mathsf{NP}}\xspace}
\newcommand{\NPC}{\ensuremath{\mathsf{NP}}\text{-complete}\xspace}
\newcommand{\NPH}{\ensuremath{\mathsf{NP}}\text{-hard}\xspace}
\newcommand{\PNPH}{para-\ensuremath{\mathsf{NP}\text{-hard}}\xspace}
\newcommand{\el}{\ensuremath{\ell}\xspace}
\newcommand{\suc}{\ensuremath{\succ}\xspace}
\newcommand{\sucb}{\ensuremath{\succ_{\text{best}}}\xspace}
\newcommand{\sucw}{\ensuremath{\succ_{\text{worst}}}\xspace}
\newcommand{\PNENP}{\ensuremath{\mathsf{P\ne NP}}\xspace}
\newcommand{\WOH}{\ensuremath{\mathsf{W[1]}}-hard\xspace}
\newcommand{\WO}{\ensuremath{\mathsf{W[1]}}\xspace}
\newcommand{\WT}{\ensuremath{\mathsf{W[2]}}\xspace}
\newcommand{\WOC}{\ensuremath{\mathsf{W[1]}}-complete\xspace}
\newcommand{\MEAF}{{\sc Minimum Edge Activation Flow}\xspace}
\newcommand{\DTASf}{{\sc Decoupled Two-Stage Allocation Strategy}\xspace}
\newcommand{\DTAS}{{\sc DTAS}\xspace}

\newcommand{\CARLf}{{\sc Capacity-Aware Reuse-first Layered Allocation}\xspace}
\newcommand{\CARL}{{\sc CARL}\xspace}

\newcommand{\WTH}{\ensuremath{\mathsf{W[2]}}-hard\xspace}
\newcommand{\FPT}{\ensuremath{\mathsf{FPT}}\xspace}
\newcommand{\fpt}{\ensuremath{\mathsf{FPT}}\xspace}
\newcommand{\xp}{\ensuremath{\mathsf{XP}}\xspace}
\newcommand{\tsat}{\ensuremath{(3,\text{B}2)}-{\sc SAT}\xspace}
\newcommand{\SAT}{\ensuremath{3}-{\sc SAT}\xspace}
\newcommand{\GFAC}{{\sc General Factor}\xspace}
\newcommand{\Pb}{\ensuremath{P}\xspace}
\newcommand{\psne}{{\sc Exists-PSNE}\xspace}
\newcommand{\bnpg}{\psne}
\newcommand{\sbnpg}{{\sc PSNE-Strict-BNPG}\xspace}
\newcommand{\vc}{\text{vc(\ensuremath{\GG})}}
\newcommand{\vdrs}{{\sc Vertex Deletion to Regular Subgraph}\xspace}
\newcommand{\ds}{{\sc Dominating Set}\xspace}
\newcommand{\TSAT}{\ensuremath{(3,\text{B}2)}-{\sc SAT}\xspace}
\newcommand{\XTC}{{\sc X3C}\xspace}

\newcommand{\DBON}{{\sc \$Bribery over Network}\xspace}
\newcommand{\SBON}{{\sc Shift Bribery over Social Network}\xspace}
\newcommand{\DSP}{{\sc Dominating set}\xspace}
\newcommand{\ktDSP}{{($k$,$t$)-\sc Dominating set}\xspace}

\newcommand{\CSB}{{\sc Combinatorial Shift Bribery}\xspace}

\let\oldlambda\lambda
\renewcommand{\lambda}{\ensuremath{\oldlambda}\xspace}
\let\oldalpha\alpha
\renewcommand{\alpha}{\ensuremath{\oldalpha}\xspace}
\let\oldDelta\Delta
\renewcommand{\Delta}{\ensuremath{\oldDelta}\xspace}

\newcommand{\tw}{\text{tw}\xspace}
\newcommand{\YES}{{\tt yes}\xspace}
\newcommand{\NO}{{\tt no}\xspace}
\newcommand{\yes}{{\tt yes}\xspace}
\newcommand{\no}{{\tt no}\xspace}
\newcommand{\true}{\text{{\sf true}}\xspace}
\newcommand{\false}{\text{{\sf false}}\xspace}

\renewcommand{\AA}{\ensuremath{\mathcal A}\xspace}
\newcommand{\BB}{\ensuremath{\mathcal B}\xspace}
\newcommand{\CC}{\ensuremath{\mathcal C}\xspace}
\newcommand{\DD}{\ensuremath{\mathcal D}\xspace}
\newcommand{\EE}{\ensuremath{\mathcal E}\xspace}
\newcommand{\FF}{\ensuremath{\mathcal F}\xspace}
\newcommand{\GG}{\ensuremath{\mathcal G}\xspace}
\newcommand{\HH}{\ensuremath{\mathcal H}\xspace}
\newcommand{\II}{\ensuremath{\mathcal I}\xspace}
\newcommand{\JJ}{\ensuremath{\mathcal J}\xspace}
\newcommand{\KK}{\ensuremath{\mathcal K}\xspace}
\newcommand{\LL}{\ensuremath{\mathcal L}\xspace}
\newcommand{\MM}{\ensuremath{\mathcal M}\xspace}
\newcommand{\NN}{\ensuremath{\mathcal N}\xspace}
\newcommand{\OO}{\ensuremath{\mathcal O}\xspace}
\newcommand{\PP}{\ensuremath{\mathcal P}\xspace}
\newcommand{\QQ}{\ensuremath{\mathcal Q}\xspace}
\newcommand{\RR}{\ensuremath{\mathcal R}\xspace}
\renewcommand{\SS}{\ensuremath{\mathcal S}\xspace}
\newcommand{\TT}{\ensuremath{\mathcal T}\xspace}
\newcommand{\UU}{\ensuremath{\mathcal U}\xspace}
\newcommand{\VV}{\ensuremath{\mathcal V}\xspace}
\newcommand{\WW}{\ensuremath{\mathcal W}\xspace}
\newcommand{\XX}{\ensuremath{\mathcal X}\xspace}
\newcommand{\YY}{\ensuremath{\mathcal Y}\xspace}
\newcommand{\ZZ}{\ensuremath{\mathcal Z}\xspace}

\newcommand{\SSS}{\overline{\SS}\xspace}

\newcommand{\MF}{\ensuremath{\mathfrak M}\xspace}
\newcommand{\AF}{\ensuremath{\mathfrak A}\xspace}
\newcommand{\GF}{\ensuremath{\mathfrak G}\xspace}
\newcommand{\PF}{\ensuremath{\mathfrak P}\xspace}

\newcommand{\ov}[1]{\ensuremath{\overline{#1}}}

\newcommand{\aaa}{\ensuremath{\mathfrak a}\xspace}
\newcommand{\bbb}{\ensuremath{\mathfrak b}\xspace}
\newcommand{\ccc}{\ensuremath{\mathfrak c}\xspace}
\newcommand{\ddd}{\ensuremath{\mathfrak d}\xspace}
\newcommand{\eee}{\ensuremath{\mathfrak e}\xspace}
\newcommand{\fff}{\ensuremath{\mathfrak f}\xspace}
\newcommand{\iii}{\ensuremath{\mathfrak i}\xspace}
\newcommand{\jjj}{\ensuremath{\mathfrak j}\xspace}
\newcommand{\kkk}{\ensuremath{\mathfrak k}\xspace}
\newcommand{\mmm}{\ensuremath{\mathfrak m}\xspace}
\newcommand{\nnn}{\ensuremath{\mathfrak n}\xspace}
\newcommand{\ooo}{\ensuremath{\mathfrak o}\xspace}
\newcommand{\ppp}{\ensuremath{\mathfrak p}\xspace}
\newcommand{\qqq}{\ensuremath{\mathfrak q}\xspace}
\newcommand{\rrr}{\ensuremath{\mathfrak r}\xspace}
\newcommand{\sss}{\ensuremath{\mathfrak s}\xspace}
\newcommand{\ttt}{\ensuremath{\mathfrak t}\xspace}
\newcommand{\uuu}{\ensuremath{\mathfrak u}\xspace}
\newcommand{\vvv}{\ensuremath{\mathfrak v}\xspace}
\newcommand{\www}{\ensuremath{\mathfrak w}\xspace}
\newcommand{\xxx}{\ensuremath{\mathfrak x}\xspace}
\newcommand{\yyy}{\ensuremath{\mathfrak y}\xspace}
\newcommand{\zzz}{\ensuremath{\mathfrak z}\xspace}

\newcommand{\EB}{\ensuremath{\mathbb E}\xspace}
\newcommand{\NB}{\ensuremath{\mathbb N}\xspace}
\newcommand{\ZB}{\ensuremath{\mathbb Z}\xspace}
\newcommand{\RB}{\ensuremath{\mathbb R^+}\xspace}

\newcommand{\ffrac}{\flatfrac}
\newcommand{\nfrac}{\nicefrac}

\newtheorem{observation}{\sc Observation}[section]
\newtheorem{claim}{\sc Claim}[section]
\newtheorem{reductionrule}{\sc Reduction rule}
\newtheorem{probdefinition}{\sc Problem Definition}[section]

\newcommand{\eps}{\ensuremath{\varepsilon}\xspace}
\renewcommand{\epsilon}{\eps}

\newcommand{\ignore}[1]{}

\newcommand{\pr}{\ensuremath{\prime}}
\newcommand{\prr}{\ensuremath{{\prime\prime}}}

\renewcommand{\leq}{\leqslant}
\renewcommand{\geq}{\geqslant}
\renewcommand{\ge}{\geqslant}
\renewcommand{\le}{\leqslant}

\newcommand{\greedy}{\textsc{Greedy committee}\xspace}

\crefname{theorem}{Theorem}{\bf Theorems}
\crefname{observation}{Observation}{\bf Observations}
\crefname{lemma}{Lemma}{\bf Lemmata}
\crefname{corollary}{Corollary}{\bf Corollaries}
\crefname{proposition}{Proposition}{\bf Propositions}
\crefname{definition}{Definition}{\bf Definitions}
\crefname{claim}{Claim}{\bf Claims}
\crefname{reductionrule}{Reduction rule}{\bf Reduction rules}
\crefname{probdefinition}{Problem Definition}{Problem Definitions}

\renewcommand{\shortauthors}{Hota et al.}
\newtheorem{axiom}[theorem]{Axiom}
\theoremstyle{remark}
\newtheorem{remark}[theorem]{Remark}
\begin{abstract}
The Unified Payments Interface (UPI) is the world's largest real-time payment system, processing more than 21 billion transactions per month, and has emerged as India's primary mode of digital payments. However, the concentration of UPI  transactions within just two apps - PhonePe and Google Pay - has raised concerns of duopoly in India's digital financial ecosystem. To address this, the National Payments Corporation of India (NPCI) has mandated that no single UPI app should exceed 30\% of total transaction volume. Enforcing this cap, however, presents a significant computational and operational challenge: {\it how to redistribute user transactions across apps %without causing widespread 
while minimizing user inconvenience and maintaining capacity limits?}

In this paper, we formalize this challenge as the {\sc Cap-Constrained UPI} problem and model it as a {\sc Minimum Edge Activation Flow} problem on a bipartite user-application network, where activating an edge represents an additional app installation. Our goal is to determine the smallest set of new installations that allows the system to route all transactions without violating app capacity constraints. We show that \MEAF is NP-Complete through a reduction from 3-{\sc Partition} problem, motivating the need for scalable approximation strategies. To tackle this challenge, we propose \DTASf (\DTAS), a scalable allocation framework with both offline and adaptive online variants. We further extend it with a fairness-aware mechanism that balances transaction load across applications while keeping additional installations low. Our experimental results on a dataset of 100~million transactions
show that DTAS comes within 1--2 additional installations of the ILP optimum at every tested scale, while running several orders of
magnitude faster. We further demonstrate a concrete trade-off curve between fairness
and installation cost that can directly inform NPCI's forthcoming regulatory enforcement deadline. To our knowledge, this is the first work addressing the problem of balancing financial transaction flows for regulatory compliance and likely have significant real-world impact.

\end{abstract}





\maketitle

\section{Introduction}
Digital payment platforms have transformed how billions of people transact. In India, this shift is driven by the Unified Payments Interface (UPI), a government-backed system that enables instant bank-to-bank transfers across more than 77 apps such as PhonePe, Google Pay, and Paytm~\cite{upi_gov}. Unlike proprietary systems, UPI is shared public infrastructure with government-subsidized zero-fee transactions, and currently accounting for over 85\% of India’s digital payments~\cite{govt-incentive}. 

In terms of absolute numbers, as of mid-2025, UPI served over 491~million users and 65~million merchants, processing over 640~million transactions daily, surpassing Visa's daily volume~\cite{cnbc2025}. 
UPI now accounts for nearly 50\% of all global real-time payment transactions, making it the world's most consequential retail payment
rail~\cite{imf2025}. Beyond India, UPI has been deployed in nine countries: Bhutan, Nepal, Singapore, Sri Lanka, Mauritius, UAE, France, Cyprus, and Qatar, and has been integrated into the PayPal World platform for
international transactions~\cite{cnbc2025}.

Despite the open architecture, usage is highly concentrated. PhonePe and Google Pay together process over 80\% of transactions~\cite{moneycontrol}. This concentration creates systemic risk, since failures in a single app can disrupt a large fraction of payments, and leads to an uneven use of public subsidies, with most benefits going to a small number of platforms. 
To address this, India's payments regulator, the National Payments Corporation of India (NPCI), introduced a policy mandating that no single UPI app may process more than 30\% of total 
transactions~\cite{reuters2024}. The goal is not to penalize dominant apps, but to ensure the ecosystem remains competitive, resilient, and equitable so that public infrastructure serves all participants, not just the most popular ones. However, enforcing this cap is far from straightforward as  users naturally gravitate toward familiar apps. NPCI can not ask  millions of users to change their habits overnight.

A naive enforcement mechanism is a tail-drop policy, blocking further transactions once an app exceeds its quota. Although effective, it leads to abrupt failures and poor user experience. A more practical solution is an alert-based strategy, where users are notified when an app nears its limit and encouraged to switch to another app. This idea parallels Random Early Detection (RED) in network congestion control~\cite{mahawish2022improving}, where early warnings prevent sharp throughput drops.

% However, implementing such alert-based strategies is fundamentally limited by user behavior: \textcolor{blue}{users are often unwilling to install and maintain multiple UPI applications}, and alerts alone may not induce sufficient switching. 
However, implementing such alert-based strategies is limited in practice, since users cannot be expected to install and maintain multiple UPI applications. As a result, routing decisions should consider user preferences and usage frequency rather than assuming constant switching across apps.
In practice, achieving compliance with the cap requires that some users install additional applications in advance, enabling the system to redistribute load away from heavily used apps. Here, our notion of fairness is not based on equilibrium considerations, but rather on preventing excessive concentration of transactions on a small subset of applications, thereby ensuring balanced utilization of the shared infrastructure. 

A natural approach is to incentivize such installations. Since the government already subsidizes UPI transactions to maintain zero fees~\cite{govt-incentive}, a portion of this budget could be repurposed to encourage users to adopt additional, underutilized applications -- for example, via time-bound or usage-linked rewards. Given that such incentives are inherently limited, the central question we address is: \textit{what is the minimum number of additional app installations required to ensure that all transactions can be routed within the regulatory cap, without any failures?}

The problem arises at two natural operational timescales in digital payment systems. In the offline setting, we assume access to historical transaction logs, which reveal user activity patterns such as frequent users and their transaction volumes. This allows the system to proactively recommend additional app installations for a carefully chosen subset of users, so that future transaction demand can be feasibly routed without disruption. Operationally, this corresponds to periodically running our augmentation algorithm on aggregated data and deploying targeted interventions (e.g., incentives) ahead of time.
In contrast, the online setting captures the real-time nature of payment systems, where transactions arrive sequentially and future demand is unknown. Upon the arrival of each transaction, the system must immediately decide how to route it and, if necessary, whether to activate a new user-application connection without knowledge of future requests.

We formalize this as the \MEAF problem on a bipartite flow network $G = (V = \{s\} \cup U \cup A \cup \{t\}, E = E_{\mathrm{solid}} \cup E_{\mathrm{dashed}})$. Here, users $U$ and apps $A$ form two partitions: solid edges represent existing app installations, dashed edges denote potential ones, $s$–$U$ edges have capacities $t_u$ (user transaction volumes), and $A$–$t$ edges have capacities $c_a$ (app limits). The objective is to activate the minimum number of dashed edges to ensure a feasible integral flow from $s$ to $t$ without exceeding any app's capacity.

Despite its intuitive formulation, the problem is computationally intractable: we prove that determining the smallest feasible activation set or equivalently, the minimal number of additional app installations is \NPC, even when restricted to only three apps. This rules out any polynomial-time exact solution unless $\mathsf{P} = \mathsf{NP}$. Consequently, we propose efficient greedy heuristics that approximate the optimal activation pattern while being scalable for large transaction networks.

\textbf{Our Contributions.} Our main contributions are as follows:

\begin{itemize}

\item We introduce the \textsc{Cap-Constrained UPI}, \textsc{Online Cap-Constrained UPI}, and \textsc{Fair Cap-Constrained UPI} problems for capacity-aware transaction routing, and show that the corresponding decision problem is \NPC\ via a reduction from \textsc{3-Partition};

\item We design two ILP formulations for the problem: an installation-minimization formulation and a fairness-aware extension for balancing transaction loads across applications;

\item We propose \DTASf, a scalable offline layered allocation strategy that exploits workload skewness and application reuse patterns to minimize unnecessary app installations while remaining close to the ILP optimum;

\item Thereafter we develop \textsc{Online}\_\DTAS, an adaptive online algorithm employing sketch-based heavy-user detection, delayed scheduling, and capacity reservation mechanisms, and further introduce \textsc{Fair\_\DTAS}, a fairness-aware extension based on composite allocation scoring; and

\item We present an extensive experimental evaluations on synthetic and transaction-driven datasets against multiple baselines, demonstrating improvements in installation efficiency, fairness, and scalability.

\end{itemize}

\longversion{
\section{Background and Related Work}

The Unified Payments Interface (UPI), developed by the National Payments Corporation of India (NPCI) under the guidance of the Reserve Bank of India (RBI), has become a major part of India’s digital economy. It provides an interoperable payment infrastructure that enables secure, real-time fund transfers and supports innovation through its open design~\cite{rastogi2021unified}. As of 2024, UPI serves over 350 million users, connects more than 550 banks, and supports 77 payment apps including Google Pay, PhonePe, BHIM, and WhatsApp Pay~\cite{upirate}. In 2023, it processed 117 billion transactions worth USD 2.19 trillion~\cite{upirate}. However, rising transaction volumes have simultaneously increased the risk of fraud and anomaly. Several scholarly works, such as~\cite{11170880, gambo2022convolutional, lingareddy2025enhancing}, have emphasized AI-driven approaches to strengthen fraud detection mechanisms in digital payment systems.

Beyond fraud detection, another critical concern in the financial domain is fairness in decision-making. Unfair decision-making in financial services occurs when certain groups or individuals face biased treatment in areas such as loan approvals, credit scoring, or mortgage access. Such bias can arise from historical discrimination or from machine learning models that inadvertently learn unfair patterns from data. Consequently, minority groups or equally qualified applicants may be denied fair financial opportunities, reinforcing existing social and economic inequalities. Song et al.~\cite{song2023towards} address this issue by proposing a Temporal Fair Graph Neural Network (TF-GNN) framework that models financial transactions as dynamic networks and enforces individual fairness over time. They introduce two new fairness notions specific to temporal graphs, provide a theoretical analysis of their fairness regret, and demonstrate through real-world experiments that their approach improves both prediction accuracy and fairness compared to prior methods.

There is another growing problem: UPI usage is highly concentrated. In September 2025, two third-party apps handled over 80\% of all UPI transactions~\cite{moneycontrol}. Such concentration creates systemic risks and limits competition. Similar concentration patterns have also been observed in other large-scale digital payment ecosystems. In China, Alipay and WeChat Pay have become the dominant digital payment platforms and largely displaced traditional card-based payments through strong network effects and ecosystem integration~\cite{cgap_china}. Studies and industry analyses report that the Chinese mobile payment market has evolved into a near-duopoly, raising concerns regarding market concentration and platform dependence~\cite{wechat_alipay_dominance}. These observations suggest that payment concentration is not unique to UPI but represents a broader challenge in digital payment systems.

To address this issue, NPCI introduced a 30\% transaction cap for third-party applications. Existing work mainly focuses on policy and market effects, with little attention to algorithmic approaches for enforcing such caps efficiently.

A natural perspective for our problem comes from network flow and augmentation problems. In particular, the {\sc Pure Fixed Charge Transportation} problem~\cite{ZHU2025100875,adlakha1999fixed,billheimer1973network,estan2021allocate} is closely related and provides important technical foundations. However, our setting differs fundamentally because activating additional edges corresponds to installing new user–application connections. This introduces behavioral and infrastructural constraints that traditional augmentation models do not capture.

Finally, our motivation for early alerts draws inspiration from congestion-control mechanisms such as Random Early Detection (RED)~\cite{mahawish2022improving}, where systems proactively react before hard capacity limits are reached. Rather than dropping requests after exceeding thresholds, RED uses early interventions to prevent congestion collapse. We adopt a similar philosophy: proactively recommending app installations before capacities become binding.

To the best of our knowledge, this is the first work that formally studies and models this concentration issue, proposing mechanisms to enhance resilience and promote a more balanced and inclusive UPI ecosystem.
}

\section{Problem Definition}

We model the problem of fairly balancing UPI transactions as a bipartite flow network. Let $U$ denote the set of users and $A$ denote the set of UPI apps. Each user $u \in U$ generates $t_u$ transactions, while each app $a \in A$ can process at most $c_a$ transactions, typically defined as a fraction of the total transaction volume. Users may already have certain apps installed, represented by solid edges $E_{\rm solid} \subseteq U \times A$, while potential additional installations are represented by dashed edges $E_{\rm dashed} \subseteq U \times A$. A source node $s$ is connected to all users with edges of capacity $t_u$, and a sink node $t$ is connected to all apps with edges of capacity $c_a$. The goal is to identify the minimal subset of dashed edges $E' \subseteq E_{\rm dashed}$ that need to be activated such that all transactions can be routed from $s$ to $t$ without exceeding app capacities or dropping transactions. Figure \ref{fig:prblem-model} illustrates our construction.

We now formally state our problems.

\begin{definition}[\textsc{Cap-Constrained UPI}]
Given a bipartite graph $G = (U, A, E_{\rm solid} \cup E_{\rm dashed})$, where an edge $(u,a) \in E_{\rm solid}$ indicates that user $u$ has application $a$ installed, and $(u,a) \in E_{\rm dashed}$ is a potential installation. Each user $u \in U$ has $t_u$ transactions to route, and each application $a \in A$ has a capacity $c_a \in \mathbb{N}$. The goal is to find a minimum subset $E' \subseteq E_{\rm dashed}$ such that all transactions can be feasibly routed via $E_{\rm solid} \cup E'$, with the load on every application $a \in A$ not exceeding $c_a$.
\end{definition}

\begin{definition}[\textsc{Online Cap-Constrained UPI}]
Given a bipartite graph $G = (U, A, E_{\rm solid} \cup E_{\rm dashed})$, where an edge $(u,a) \in E_{\rm solid}$ indicates that user $u$ has application $a$ installed, and $(u,a) \in E_{\rm dashed}$ is a potential installation. Each application $a \in A$ has a capacity $c_a \in \mathbb{N}$. Transactions arrive as an online sequence $\sigma = (\tau_1, \tau_2, \ldots, \tau_n)$, where each $\tau_i$ belongs to some user $u_i \in U$ and must be routed immediately and irrevocably. The goal is to find a minimum subset $E' \subseteq E_{\rm dashed}$ such that all transactions can be feasibly routed via $E_{\rm solid} \cup E'$, with the load on every application $a \in A$ not exceeding $c_a$.
\end{definition}

\begin{definition}[\textsc{Fair Cap-Constrained UPI}]
Given a bipartite graph $G = (U, A, E_{\rm solid} \cup E_{\rm dashed})$, where an edge $(u,a) \in E_{\rm solid}$ indicates that user $u$ has application $a$ installed, and $(u,a) \in E_{\rm dashed}$ is a potential installation. Each user $u \in U$ has $t_u$ transactions to route, and each application $a \in A$ has a capacity $c_a \in \mathbb{N}$. The goal is to find a minimum subset $E' \subseteq E_{\rm dashed}$ such that all transactions can be feasibly routed via $E_{\rm solid} \cup E'$, with the load on every application $a \in A$ not exceeding $c_a$, and the load vector $L = (\ell_1, \ldots, \ell_{|A|})$ is balanced across applications.
\end{definition}

Let \(L = (\ell_1, \ell_2, \ldots, \ell_{|A|})\)
denote the vector of application loads. Fairness of a routing is evaluated using:
\begin{itemize}
    \item \textbf{Max-Min Gap:}
    \[
    \max_{a \in A} \ell_a
    -
    \min_{a \in A} \ell_a.
    \]
    Smaller gap indicates fairer load balancing.

    \item \textbf{Gini Coefficient:}
    \[
    G(L)
    =
    \frac{
    \sum_{i=1}^{|A|}
    \sum_{j=1}^{|A|}
    |\ell_i - \ell_j|
    }{
    2|A| \sum_{i=1}^{|A|} \ell_i
    }.
    \]
    Lower Gini coefficient implies a more equal distribution of transactions across applications.
\end{itemize}

The \textsc{Fair Cap-Constrained UPI} problem asks for a feasible routing that minimizes
the number of activated edges $|E'|$ while simultaneously maintaining balanced
application loads according to the above fairness measures.
%\end{definition}

\section{Complexity Status}
 We formulate the problem theoretically as the {\sc Minimum Edge Activation Flow (MEAF)} problem. 
%Formally, let $f(u,a)$ denote the number of transactions routed from user $u$ to app $a$. The flow must satisfy the following conditions: every user’s transactions are fully routed, i.e., $\sum_{a\in A} f(u,a) = t_u$, and the total flow into any app respects its capacity, $\sum_{u\in U} f(u,a) \le c$. Flow is allowed only on solid edges or activated dashed edges, and transaction units are indivisible, so $f(u,a) \in \mathbb{Z}_{\ge 0}$. Finally, the total flow in the network equals the total number of transactions, ensuring no transaction is dropped. The objective is to minimize $|E'|$, the number of additional app installations required to achieve a feasible integral flow. We now formally show that \MEAF is \NPC.

\begin{center}
\fbox{%
\parbox{0.95\columnwidth}{%
\MEAF (MEAF)

\textbf{Input.}
Given a bipartite graph $G = (U \cup A, E_{\mathrm{solid}} \cup E_{\mathrm{dashed}})$, 
a source node $s$ connects to each $u \in U$ by an edge of capacity $t_u$, 
and a sink node $t$ connects to each $a \in A$ by an edge of capacity $c$. 
Edges in $E_{\mathrm{solid}}$ represent existing connections, while 
$E_{\mathrm{dashed}}$ denote optional edges that can be activated.

\textbf{Output:} Find the smallest subset $E' \subseteq E_{\mathrm{dashed}}$ 
such that all demands $t_u$ can be routed from $s$ to $t$ through 
$E_{\mathrm{solid}} \cup E'$ without violating any capacity $c$.}
}
\end{center}

%\longversion{
\begin{figure}
    \centering
  \includegraphics[width=0.45\textwidth]{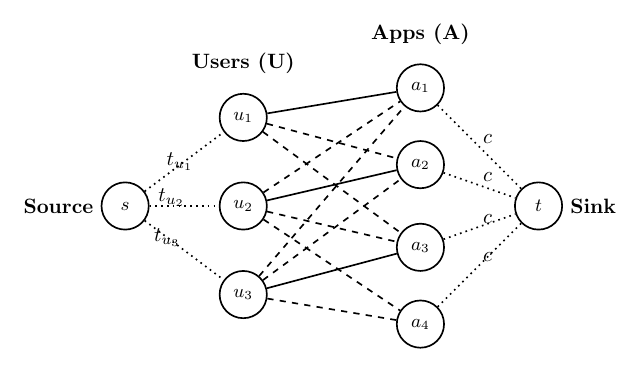}
    \caption{Illustration of the bipartite UPI transaction flow network. Solid edges represent \textbf{pre-installed apps} $(E_{\text{solid}})$, dashed edges denote \textbf{potential installations} $(E_{\text{dashed}})$, 
and dotted edges indicate \textbf{user transactions} $(t_{u_i})$ 
and \textbf{app capacities} $(c)$.}
    \label{fig:prblem-model}
\end{figure}
%}
Formally, let $f(u,a)$ denote the number of transactions routed from user $u$ to app $a$. The flow must satisfy the following conditions: every user’s transactions are fully routed, i.e., $\sum_{a\in A} f(u,a) = t_u$, and the total flow into any app respects its capacity, $\sum_{u\in U} f(u,a) \le c$. Flow is allowed only on solid edges or activated dashed edges, and transaction units are indivisible, so $f(u,a) \in \mathbb{Z}_{\ge 0}$. Finally, the total flow in the network equals the total number of transactions, ensuring no transaction is dropped. The objective is to minimize $|E'|$, the number of additional app installations required to achieve a feasible integral flow. We now formally show that \MEAF is \NPC.

\begin{theorem}
\MEAF is \NPC.
\end{theorem}
\shortversion{
\begin{proofsketch}
It is easy to see that the problem is in NP. We prove NP-hardness by a polynomial-time reduction from \textsc{3-Partition}.
Let an instance of \textsc{3-Partition} be given by $S=\{s_1,\dots,s_{3m}\}$ and $B$ with $\sum_i s_i=mB$ and $B/4<s_i<B/2$ for all $i$.
We construct an instance of \textsc{UPI flow problem} as follows. Create $3m$ users $U=\{u_1,\dots,u_{3m}\}$ with $t_{u_i}=s_i$, $\forall i \in [3m]$.
Create $m$ apps $A=\{a_1,\dots,a_m\}$ with capacities $c_{a_j}=B$ for all $j$.
Add source $s$ and sink $t$ with edges $(s,u_i)$ of capacity $t_{u_i}$ and edges $(a_j,t)$ of capacity $B$.
Let $E_{\mathrm{solid}}=\varnothing$ and $E_{\mathrm{dashed}}=U\times A$ (i.e., every user can connect to every app via a dashed edge). All $(u,a)$ edges have infinite capacity so that the app capacities are the binding constraints. Set the activation budget $k:=3m$. 
%This construction is clearly polynomial in the input size. 
We then argue that the \textsc{3-Partition} instance is a YES-instance if and only if the constructed flow instance admits a feasible integral flow using at most $k=3m$ activated dashed edges. Due to space constraints, the detailed proof is presented in the extended version of the paper\footnote{Available at \href{https://github.com/ENAMINE1/FDP}{https://github.com/ENAMINE1/FDP}}.%\cite{hota2025fdp}.
\end{proofsketch}
}

\longversion{
\begin{proof}
It is easy to see that the problem is in NP: given a set of activated dashed edges and an integral flow assignment, we can verify in polynomial time that (i) flow conservation holds, (ii) capacities on $(s,u)$ and $(a,t)$ are respected, (iii) flow is sent only on activated dashed edges, and (iv) the number of activated dashed edges is at most $k$.

We prove NP-hardness by a polynomial-time reduction from \textsc{3-Partition}.
Let an instance of \textsc{3-Partition} be given by $S=\{s_1,\dots,s_{3m}\}$ and $B$ with $\sum_i s_i=mB$ and $B/4<s_i<B/2$ for all $i$.
We construct an instance of \textsc{UPI flow problem} as follows.

\smallskip
\noindent\emph{Construction.}
Create $3m$ users $U=\{u_1,\dots,u_{3m}\}$ with $t_{u_i}=s_i$, $\forall i \in [3m]$.
Create $m$ apps $A=\{a_1,\dots,a_m\}$ with capacities $c_{a_j}=B$ for all $j$.
Add source $s$ and sink $t$ with edges $(s,u_i)$ of capacity $t_{u_i}$ and edges $(a_j,t)$ of capacity $B$.
Let $E_{\mathrm{solid}}=\emptyset$ and $E_{\mathrm{dashed}}=U\times A$ (i.e., every user can connect to every app via a dashed edge).
All $(u,a)$ edges have sufficiently large capacity (e.g., capacity $t_u$) so that the app capacities are the binding constraints. Set the activation budget $k:=3m$.

This construction is clearly polynomial in the input size.

\smallskip
\noindent\emph{(\,$\Rightarrow$\,) If the \textsc{3-Partition} instance is a YES-instance, then the constructed flow instance admits a feasible integral flow using at most $k=3m$ activated dashed edges.}
Assume $s$ can be partitioned into $m$ disjoint triples $T_1,\dots,T_m$ with $\sum_{s\in T_j} s=B$ for each $j$.
For each triple $T_j$, choose an app $a_j$ and for each user $u_i$ with $s_i\in T_j$ activate exactly the dashed edge $(u_i,a_j)$ and send the full amount $t_{u_i}=s_i$ on that edge.
For each app $a_j$, the incoming flow is $\sum_{u_i\in T_j} t_{u_i}=B$, so the capacity $(a_j,t)$ is respected.
Each user uses exactly one dashed edge, so the number of activated dashed edges is exactly $3m=k$.
All source capacities $(s,u_i)$ are saturated, hence the total flow value is $\sum_i s_i=mB$.
Thus there exists a feasible integral flow using at most $k$ activations.

\smallskip
\noindent\emph{(\,$\Leftarrow$\,) If the constructed flow instance admits a feasible integral flow using at most $k=3m$ activated dashed edges, then the \textsc{3-Partition} instance is a YES-instance.}
Suppose there is a feasible integral $s$--$t$ flow of value $\sum_{i=1}^{3m} s_i=mB$ using at most $3m$ activated dashed edges.
Since $E_{\mathrm{solid}}=\varnothing$ and every user $u_i$ has positive demand $t_{u_i}=s_i>0$, each user must have at least one activated outgoing dashed edge in order to send any flow.
Therefore any feasible solution uses at least $3m$ activated dashed edges.
By the budget bound, the solution uses exactly $3m$ activations, hence each user activates \emph{exactly one} dashed edge and sends all of its (integral) demand through that edge.

Let $S_j\subseteq U$ be the set of users assigned to app $a_j$ (i.e., those for which $(u,a_j)$ is activated and carries the full $t_u$).
Flow feasibility and capacity imply for every $j$ that
\[
\sum_{u\in S_j} t_u \;\le\; c_{a_j} \;=\; B.
\]
Summing over all apps and using that the total flow equals $mB$ gives
\[
\sum_{j=1}^m \sum_{u\in S_j} t_u \;=\; \sum_{u\in U} t_u \;=\; mB \;=\; \sum_{j=1}^m B,
\]
which forces equality in each app separately: $\sum_{u\in S_j} t_u = B$ for all $j$.
Finally, by the \textsc{3-Partition} bounds $B/4<t_u<B/2$, no app can receive $1$ item (any $t_u<B/2$) or $\ge 4$ items (each $t_u>B/4$ would exceed $B$).
Therefore each $S_j$ has exactly three users and their demands sum to $B$.
The family $\{S_1,\dots,S_m\}$ thus yields a partition of $s$ into $m$ triples each summing to $B$, i.e., a YES-solution for \textsc{3-Partition}.
\end{proof}
}
% For each element $e_i$ create a user $u_i$ with demand $t_{u_i}=1$, and for each set $S_j$ create an app $a_j$. Add a dashed edge $(u_i,a_j)$ iff $e_i\in S_j$, connect $s$ to every $u_i$ with capacity $1$, and the edges $(a_j,t)$ each with capacity equals to $|S_j|$. A non zero integral flow through $(a_j,t)$ corresponds to selecting $S_j$, so a minimum activation set that routes all unit demands yields a minimum set cover by activating exactly $n$ dashed edges.

It is easy to see that \MEAF problem generalizes the \textsc{Set Cover} problem and hence inherits its hardness of approximation \cite{feige1998threshold}. We now state the inapproximability result.

\begin{remark}
%The \MEAF problem generalizes the classical \textsc{Set Cover} problem. Consequently, 
Unless $\mathsf{P} = \mathsf{NP}$,
\MEAF admits no polynomial-time approximation algorithm with an appoximation factor better than $(1 - o(1)) \log n$.
\end{remark}

\section{Proposed Methodology}
As shown in the previous section, our problems are \NPC. Therefore, we seek heuristics that perform well in practice. This section presents two ILP formulations for the transaction-routing problem. The first formulation minimizes additional app installations under application-capacity constraints, while the second extends the model with a max-min fairness objective over application loads. 

We then introduce \DTASf, an offline greedy allocation framework that provides key empirical insights into the trade-off between installation efficiency and fairness. Building on these observations, we finally present a family of online algorithms, including our proposed adaptive framework \textsc{Online}\_\DTAS.

\subsection{ILP-Based Routing Framework} 
We begin by presenting Integer Linear Programming (ILP) formulations for the transaction-routing problem. The formulations model the two central objectives in our setting: minimizing additional app installations and achieving fairness across applications.

\textbf{Install-Minimization ILP}: This problem can be formulated as an Integer Linear Programming (ILP) problem. We introduce variables $f(u,a) \in \mathbb{Z}_{\ge 0}$ for the number of transactions routed from user $u$ to app $a$, and binary variables $x(u,a)$ for $(u,a) \in E_{\rm dashed}$ indicating whether a potential edge is activated.

The objective is to minimize the number of additional edges activated:
\begin{equation}
\min \sum_{(u,a) \in E_{\rm dashed}} x(u,a).
\label{eq:ilp-objective}
\end{equation}

Each user's transactions must be fully routed, captured by
\begin{equation}
\sum_{a \in A} f(u,a) = t_u, \quad \forall u \in U,
\label{eq:ilp-flowconservation}
\end{equation}
while the capacities of apps must not be exceeded:
\begin{equation}
\sum_{u \in U} f(u,a) \le c_a, \quad \forall a \in A.
\label{eq:ilp-capacity}
\end{equation}

Flow through a dashed edge is permitted only if the edge is activated:
\begin{equation}
f(u,a) \le t_u \cdot x(u,a), \quad \forall (u,a) \in E_{\rm dashed}.
\label{eq:ilp-activation}
\end{equation}

Flow on solid edges is unconstrained by activation, i.e., $f(u,a) \ge 0$ for $(u,a) \in E_{\rm solid}$. All flows are integral:
\begin{equation}
f(u,a) \in \mathbb{Z}_{\ge 0}, \quad x(u,a) \in \{0,1\}.
\label{eq:ilp-integrality}
\end{equation}

\medskip

\noindent\textbf{Max-Min Fairness ILP: }For the fair variant, we keep the same routing variables and add $y_a$ for additional installations on app $a$ and $T$ for the minimum installation level.

The objective becomes:
\begin{equation}
\max T.
\label{eq:maxmin-objective}
\end{equation}

The extra constraints are:
\begin{equation}
y_a = \sum_{u \in U} x(u,a), \quad \forall a \in A,
\label{eq:maxmin-installations}
\end{equation}

\begin{equation}
y_a \ge T, \quad \forall a \in A,
\label{eq:maxmin-fairness}
\end{equation}

\begin{equation}
\sum_{a \in A} y_a \le B.
\label{eq:maxmin-budget}
\end{equation}

Here, $B$ is not a free parameter: it is set to the total number of additional app installations returned by the Install-Minimization ILP. Therefore, the second ILP does not increase the installation volume; instead, it redistributes the same installation budget across apps to maximize the minimum installation level $T$.

Together, these formulations capture both efficiency and fairness in the routing model. The first ILP determines how many new installations are minimally necessary to serve all demand, and the second uses exactly that amount as a cap and decides how to distribute installations so that the least-installed app is as high as possible.

\subsection{Offline Strategy}

Although the ILP formulations provide optimal allocations, they become computationally expensive at large scales. To improve scalability, we develop an offline greedy heuristic called \DTASf\ (Algorithm~\ref{alg:dtas}). 

In the offline setting, all transaction requests are available in advance, allowing users to be reordered prior to allocation. As shown in Algorithm~\ref{alg:dtas}, transactions for each user are allocated using a three-layer strategy:
\begin{enumerate}
    \item preinstalled applications,
    \item previously activated extra applications,
    \item new app installations only when necessary.
\end{enumerate}

By prioritizing reuse before introducing new installations, \DTASf\ reduces redundant app activations while improving capacity utilization.

A natural strategy is to process users in descending order of their transaction-to-app ratio, prioritizing heavy users first. Intuitively, allocating demanding users early appears beneficial because they have larger transaction loads and potentially fewer feasible allocation choices. However, empirical analysis on real transaction traces revealed a highly skewed workload distribution, where a small fraction of users generated the majority of transactions (see Figure~\ref{fig:user-skew}). Prioritizing these heavy users rapidly saturated frequently installed high-capacity applications, leaving lightweight users with fewer reusable options and forcing additional application installations despite sufficient aggregate capacity being available.

To address this issue, Algorithm~\ref{alg:dtas} instead processes users in ascending order of their transaction-to-app ratio, prioritizing lightweight users first. Since lightweight users constitute the majority, serving them early preserves allocation flexibility and significantly reduces unnecessary installations while maintaining scalability.

\textbf{Insights from \DTAS.}
The behavior of \DTAS revealed that the timing of heavy-user allocations has a substantial impact on system efficiency. Immediately serving heavy users reduces future reuse opportunities by consuming a large fraction of the available application capacity early in the allocation process. In contrast, postponing heavy-user requests allows lightweight users to exploit existing installations first, preserving flexibility and reducing redundant application activations.

This observation directly motivates the delay mechanism in \textsc{Online\_\DTAS}, where dynamically identified heavy users are temporarily stalled and scheduled later to improve long-term allocation efficiency.

% $A_u$: preinstalled applications of user $u$
% $P$: shared application pool
% $cap[a]$: remaining capacity of application $a$

\begin{algorithm}[t]
\caption{\DTAS: Decoupled Two-Stage Allocation Strategy}
\label{alg:dtas}
\small
\KwIn{
Users $U$, transaction demand $T_u$ for each user $u$,\\
preinstalled applications $A_u$, application capacity $c$
}
\KwOut{
Updated application assignments and transaction allocation
}

Initialize remaining capacity $cap[a]\gets c$ for all applications $a$\;

Initialize shared application pool:
\[
P \gets \bigcup_{u\in U} A_u
\]

Sort users by increasing demand and number of installed apps\;

\BlankLine
\textbf{Phase 1: Allocate using preinstalled applications}

\ForEach{$u\in U$}{
    $remaining \gets |T_u|$\;

    Sort $A_u$ by decreasing $cap[a]$\;

    \ForEach{$a\in A_u$}{
        allocate as many transactions as possible to $a$\;

        update $cap[a]$\;

        \If{$remaining =0$}{
            break
        }
    }

    store remaining demand of user $u$\;
}

\BlankLine
\textbf{Phase 2: Satisfy unmet demand}

\ForEach{$u\in U$ with remaining demand}{

    \tcp{Step 2a: borrow from existing shared apps}

    candidateApps $\gets P\setminus A_u$\;

    Sort candidateApps by decreasing $cap[a]$\;

    allocate transactions while capacity exists\;

    \BlankLine

    \tcp{Step 2b: install fresh applications if needed}

    \If{demand still remains}{

        freshApps $\gets$ apps not in $A_u$ and not in $P$\;

        Sort freshApps by decreasing $cap[a]$\;

        allocate transactions\;

        add newly installed apps to $P$\;
    }
}

Return allocation and updated application assignments\;

\end{algorithm}

\subsection{Online Strategies}
In practice, transaction-routing decisions arise in a dynamic environment where requests arrive sequentially and decisions must be made in real time. At each step, the system faces a choice: continue routing through currently available applications or activate additional applications when necessary. Since future requests are unknown, decisions must be made online using only the information available at the current time.

\begin{figure}[h]
    \centering
    \includegraphics[width=\linewidth]{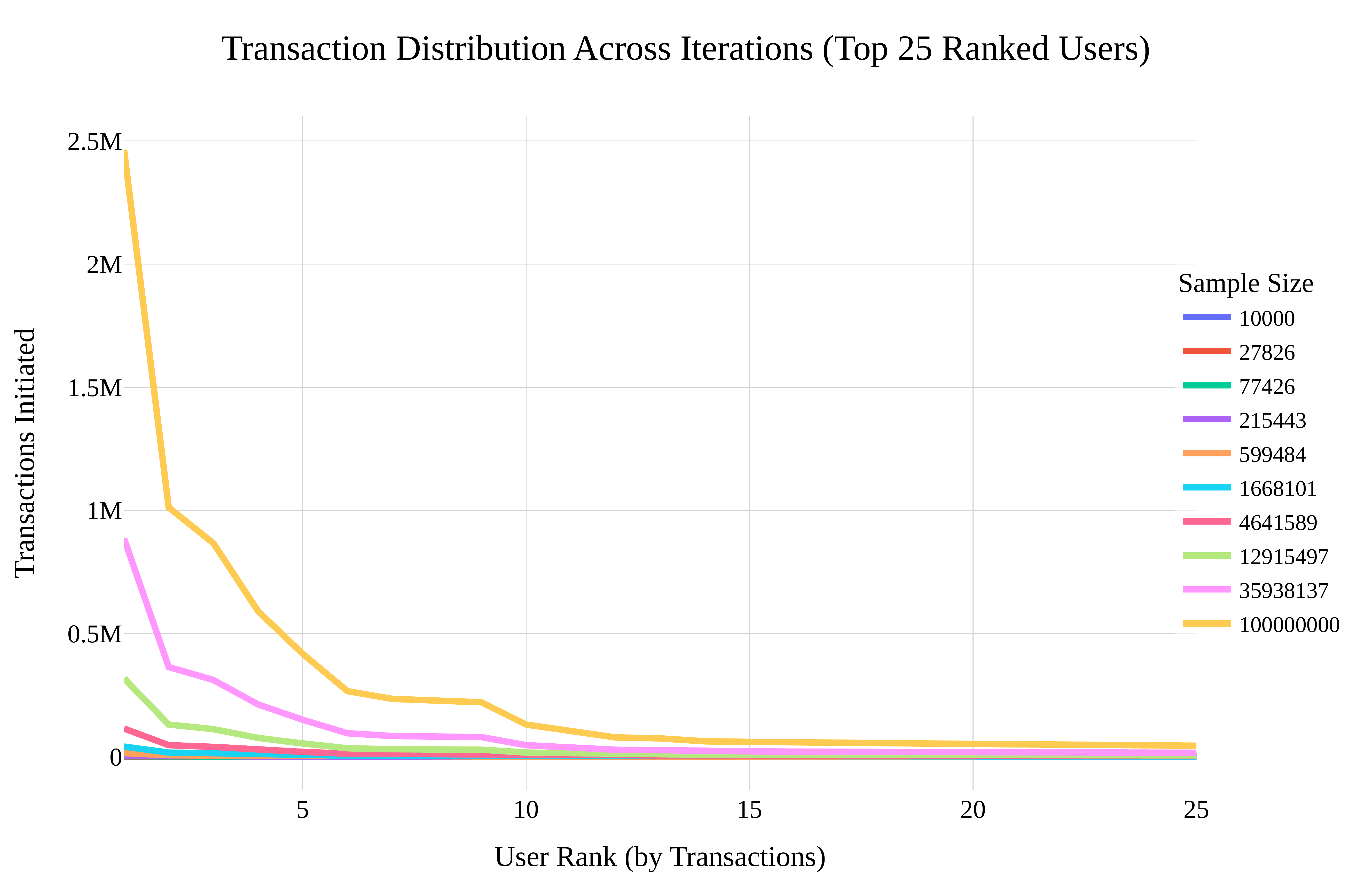}
    \caption{Distribution of user transaction volumes showing a highly skewed workload, where a small fraction of users contributes a disproportionately large number of transactions.}
    \label{fig:user-skew}
\end{figure}

Motivated by this setting, we propose two online strategies: \textsc{Online\_DTAS} and \textsc{Fair\_DTAS}.
\begin{itemize}
    \item \noindent\textbf{\textsc{Online\_\DTASh}}.
To handle streaming transaction arrivals, we further extend \DTAS into an online framework (Algorithm~\ref{alg:online-dtas}). Unlike the offline setting, complete user information is not available in advance, making it impossible to identify heavy and light users beforehand. Therefore, \textsc{Online\_\DTAS} employs an adaptive frequency-tracking mechanism based on the \textit{Space-Saving} algorithm~\cite{metwally2005spacesaving} to dynamically identify heavy users from streaming transaction data.

As illustrated in Algorithm~\ref{alg:online-dtas}, the scheduler processes a unified transaction stream and continuously updates user frequencies through a compact Space-Saving sketch. Rather than maintaining explicit counters for every user, the sketch stores only a bounded number of entries and incrementally tracks approximate transaction frequencies using limited memory. For each incoming transaction, the corresponding user frequency estimate is updated within the sketch. Periodically, sketch counters are aged using multiplicative decay and users are reclassified according to the highest estimated frequency values. Specifically, users whose sketch scores fall within the top percentile are identified as heavy users, while the remaining users are treated as light users. This mechanism enables the algorithm to adapt to evolving workload patterns while preventing stale historical activity from dominating future classifications.

Heavy-user transactions are temporarily delayed and inserted into a bounded priority queue ordered by estimated transaction frequency, ensuring that the \textit{least-heavy} delayed users are released first. In contrast, light-user requests are immediately allocated using the three-tier application selection strategy shown in Algorithm~\ref{alg:find-app}, inherited from offline \DTAS:

    \begin{algorithm}[t]
\caption{\textsc{Online\_}\DTAS: Online DTAS with Heavy-User Delaying}
\label{alg:online-dtas}
\small

\KwIn{
Transaction stream $S$, preinstalled apps $A_u$,\\
application capacity $c$
}

\KwOut{
Updated application assignments and transaction allocation
}

Initialize remaining capacities $cap[a]\gets c$\;

Reserve a fraction $\alpha c$ of each application's capacity\;

Initialize shared application pool:
\[
P\gets \bigcup_{u\in U}A_u
\]

Initialize Space-Saving sketch with at most $k$ counters\;

Initialize heavy-user set $H\gets\emptyset$\;

Initialize bounded delay queue $Q$\;

Initialize light transaction counter\;

\BlankLine

\ForEach{transaction $(u,t)\in S$}{

    Update frequency estimate of user $u$
    in the Space-Saving sketch\;

    \BlankLine

    \If{reclassification interval reached}{

        Age sketch counters using decay factor $\rho$\;

        Identify heavy users as top-percentile
        sketch entries\;

        Update heavy-user set $H$\;
    }

    \BlankLine

    \eIf{$u\in H$}{

        Insert transaction into delay queue $Q$\;

        \If{$|Q|$ exceeds threshold}{

            Drain least-heavy delayed request\;
        }

    }{

        Allocate transaction immediately using
        Algorithm~\ref{alg:find-app}\;

        Increment light transaction counter\;

        \If{drain interval reached}{

            Determine drain count based on queue size\;

            Drain delayed transactions from $Q$\;
        }
    }
}

\BlankLine

Drain all remaining delayed transactions\;

Return allocation\;

\end{algorithm}

    \begin{algorithm}[t]
    \caption{Application Selection Procedure}
    \label{alg:find-app}
    \small
    
    \KwIn{
    User $u$, installed apps $A_u$,\\
    shared pool $P$, remaining capacities
    }
    
    \KwOut{
    Selected application
    }
    
    \textbf{Tier 1:} Search installed applications\;
    
    \If{feasible app exists}{
        Return highest affinity, highest capacity app\;
    }
    
    \BlankLine
    
    \textbf{Tier 2:} Search shared-pool applications\;
    
    \If{feasible app exists}{
        Return highest affinity, highest capacity app\;
    }
    
    \BlankLine
    
    \textbf{Tier 3:} Search fresh applications\;
    
    \If{feasible app exists}{
        Return highest affinity, highest capacity app\;
    }
    
    Return NULL\;
    
    \end{algorithm}
    
\item \noindent\textbf{\textsc{Fair\_\DTASh}}:
\textsc{Fair\_\DTAS} extends the base online \DTAS\ framework by incorporating an explicit fairness objective into the application selection process (Algorithm~\ref{alg:fair-dtas}). The overall scheduling pipeline—including the streaming transaction model, heavy-user classification, delayed queue mechanism, adaptive draining policy, and capacity reservation strategy—remains identical to the base online framework described in Algorithm~\ref{alg:online-dtas}. The primary modification lies in replacing the application selection strategy with a fairness-aware allocation rule.

As shown in Algorithm~\ref{alg:fair-dtas}, the key idea behind \textsc{Fair\_\DTAS} is to balance transaction load across applications while minimizing unnecessary installations. Instead of selecting applications solely based on availability and reuse, the algorithm evaluates feasible applications using a fairness-aware composite scoring function:

\begin{equation}
Score(a)=
\alpha \cdot InstallPenalty(a)
+
\beta \cdot \Delta Fair(a),
\end{equation}

where \textit{InstallPenalty(a)} is zero for already-installed applications and positive otherwise, scaled by user--app affinity to encourage reuse. The fairness term $\Delta Fair(a)$ estimates the incremental increase in global load imbalance if the incoming transaction is assigned to application $a$. Intuitively, applications that are already heavily loaded incur larger fairness penalties, whereas lightly utilized applications receive lower penalties. In our implementation, this quantity is computed as the increase in load variance across applications after hypothetically assigning the transaction to app $a$. Consequently, assignments that disproportionately increase concentration on popular applications become less attractive.

The parameters $\alpha$ and $\beta$ determine the trade-off between minimizing installation overhead and improving fairness. Larger values of $\alpha$ favor application reuse and lower installation cost, whereas larger values of $\beta$ emphasize balanced load distribution.

The selected application is the one with minimum score; ties are broken using \textit{(i) higher user--app affinity} and \textit{(ii) larger remaining capacity}. By penalizing allocations that increase load imbalance, Algorithm~\ref{alg:fair-dtas} avoids excessive concentration on a small set of applications and promotes more balanced utilization.

By explicitly incorporating fairness into allocation decisions, \textsc{Fair\_\DTAS} achieves improved load balancing across applications while preserving the low-installation characteristics of the original \DTAS\ framework, albeit with a modest increase in installation overhead.

    \begin{algorithm}[t]
\caption{\textsc{Fair\_}\DTAS: Fairness-Aware Online DTAS}
\label{alg:fair-dtas}
\small

\KwIn{
Transaction stream $S$, preinstalled apps $A_u$,\\
application capacity $c$, fairness weights $\alpha,\beta$
}

\KwOut{
Updated application assignments and transaction allocation
}

Execute the online scheduling procedure from
Algorithm~\ref{alg:online-dtas}\;

\BlankLine

Replace the application-selection strategy with the following
fairness-aware procedure:\;

\ForEach{incoming transaction of user $u$}{

    Determine feasible applications:
    \[
    F\gets\{a \mid cap[a]>0\}
    \]

    \ForEach{$a\in F$}{

        Compute installation cost:
        
        \[
        InstallCost(a)=
        \begin{cases}
        0,& a\in A_u\\
        \frac{1}{1+\textit{Affinity}(u,a)},
        &\text{otherwise}
        \end{cases}
        \]

        Estimate fairness impact:
        
        \[
        \Delta Fairness(a)
        \]

        Compute combined score:
        
        \[
        Score(a)
        =
        \alpha \cdot InstallCost(a)
        +
        \beta \cdot \Delta Fairness(a)
        \]
    }

    Select application with minimum score\;

    Break ties using:
    
    \begin{enumerate}
        \item higher user--app affinity
        \item larger remaining capacity
    \end{enumerate}

    Assign transaction and update app load\;
}

Return allocation\;

\end{algorithm}
    \end{itemize}

\noindent

\begin{figure*}[t]
\centering

\begin{subfigure}[b]{0.4\linewidth}
    \centering
    \includegraphics[width=\linewidth]{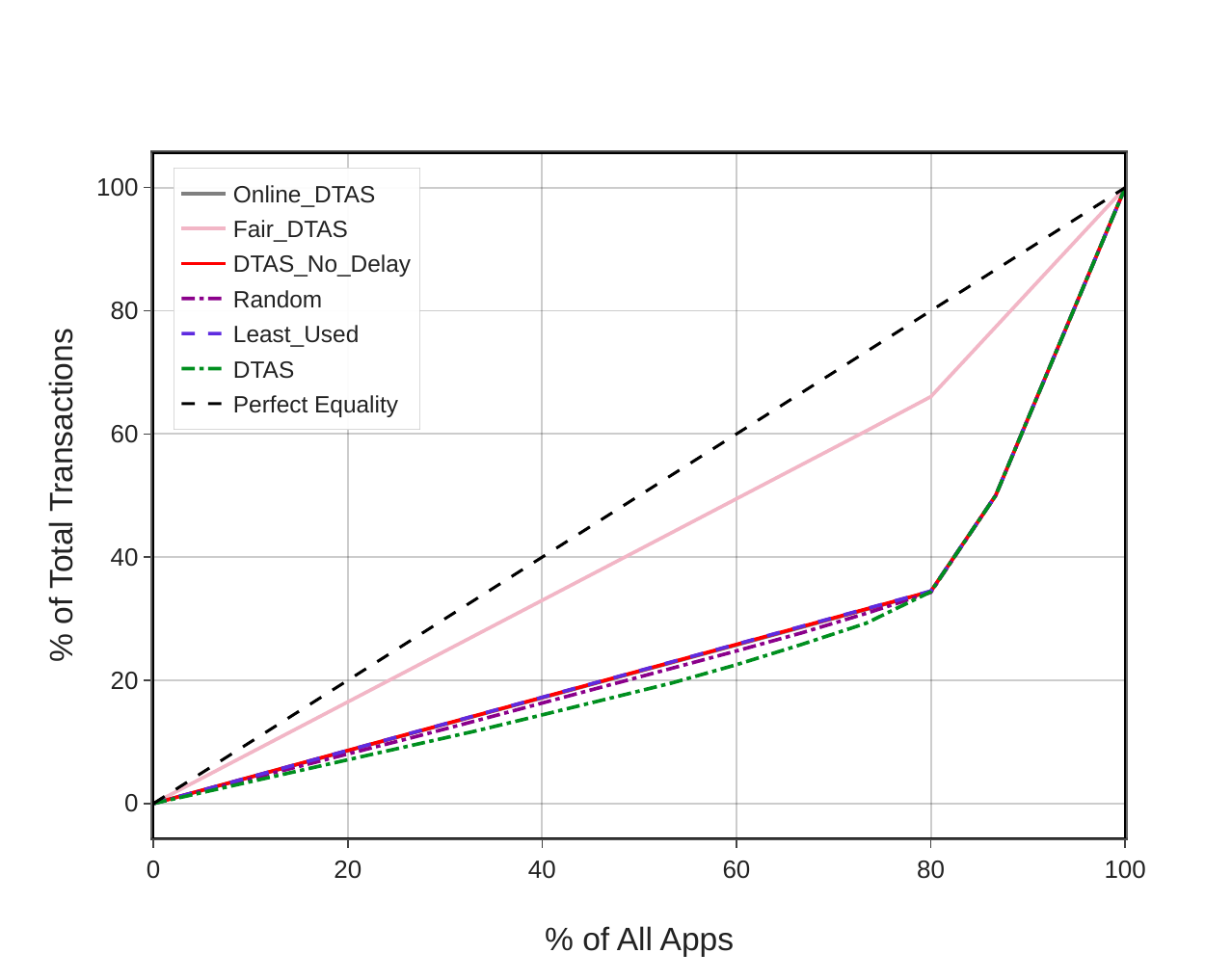}
    \caption{Lorenz curve for various algorithms.}
    \label{fig:sub1}
\end{subfigure}
\hspace{0.01\linewidth}
\begin{subfigure}[b]{0.55\linewidth}
    \centering
    \includegraphics[width=\linewidth]{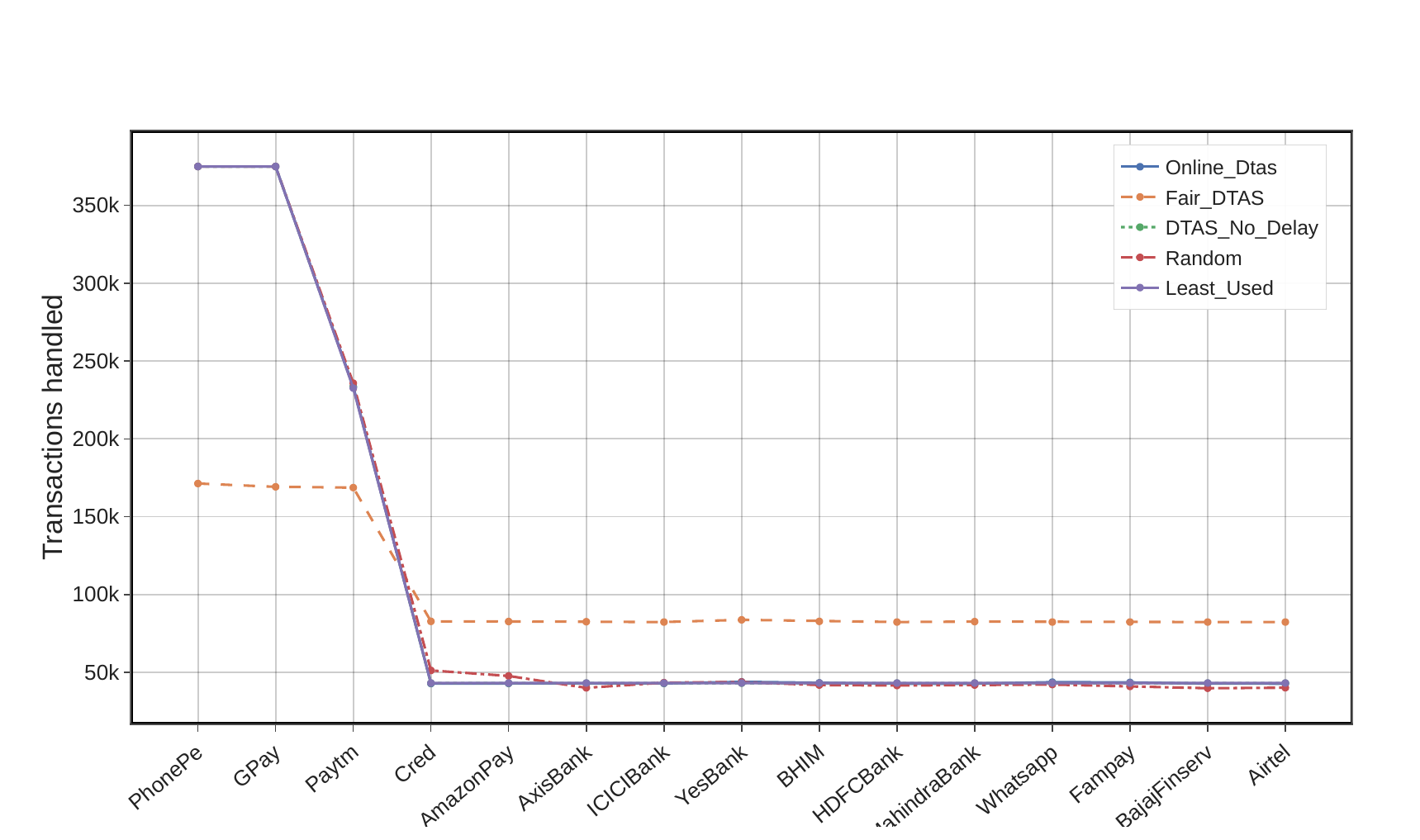}
    \caption{Load Distribution across apps.}
    \label{fig:sub2}
\end{subfigure}

\caption{Lorenz curve comparison across allocation approaches.}
\label{fig:lorenz-curve}
\end{figure*}

\section{Experimental Evaluations}
In this section, we present the experimental analysis on semi- synthetic data generated from Rabobank Transactinos data.

\textbf{Computational Environment.} We conduct experiments on a Linux-based system with an Intel(R) Xeon(R) CPU E5-2630 v2 @ 2.60\,GHz and 125\,GB RAM running Ubuntu 22.04 LTS. All algorithms are implemented in Python 3.12.3, and the Gurobi Optimizer (version 12.0.3) is used for solving the ILP formulations. Each experimental configuration represents a distinct combination of users, transaction workloads, and UPI applications under varying capacity constraints.

\textbf{Dataset Description.} Since real-world UPI transaction data is proprietary and not publicly accessible due to privacy and regulatory constraints, we use the Rabobank Transaction Dataset~\citep{rabobankdataset} as the underlying transaction corpus. The dataset contains anonymized records of banking activity within Rabobank collected over an 11-year period (2010--2020), covering 1,624,030 bank accounts and 4,127,043 transaction relationships. The original dataset includes customer-to-customer (C2C), customer-to-business (C2B), and business-to-business (B2B) interactions.

To approximate a UPI ecosystem, we assign each user a set of preinstalled payment applications according to NPCI-reported monthly transaction shares. Consequently, users are more likely to adopt high-volume applications such as PhonePe and GPay while maintaining representation of lower-volume applications. This augmentation captures realistic application adoption patterns while preserving diversity in user preferences.

Using this enriched dataset, we generate experimental instances across varying scales ranging from 10,000 to 100 million transactions and up to 1.2 million users under different transaction-to-user ratios and uniform application-capacity settings.
% Each configuration was tested under identical conditions to ensure fairness in comparison. 

% \shortversion{
%     \textbf{Experimental Setup}: Experiments were conducted on a Linux-based system with an 
% \textbf{Intel(R) Xeon(R) E5-2630 v2 @ 2.60\,GHz} CPU 
% (\textbf{12 cores, 24 threads}), \textbf{125\,GiB RAM}, and 
% \textbf{Ubuntu 22.04 LTS (64-bit)}. 
% No dedicated GPU was used. 
% All algorithms were implemented in \textsc{Python 3.12.3} and solved using 
% the \textsc{Gurobi Optimizer}~(v12.0.3).

% }

\shortversion{In this section, we present our empirical experiments and discuss the resulting performance insights.}

% In this section, we present our empirical experiments and discuss the resulting performance insights. Since UPI transaction data is not publicly available, we repurpose the Rabobank financial transaction dataset~\cite{saxena2021banking} as a proxy. The original multi-year dataset contains anonymized C2C, C2B, and B2B transfers; we preserve these proportions and uniformly redistribute the transactions over a 30-day window. To simulate UPI usage, we assign each user a set of preinstalled apps according to NPCI-reported monthly transaction shares, so that users adopt high-volume apps (e.g., PhonePe, GPay) with higher probability while still reflecting the presence of smaller apps. 
%Using this semi-synthetic dataset, we generate experiment instances of varying scales—ranging from 10,000 to 100 million transactions and up to 1.2 million users—under different transaction-to-user ratios and uniform app-capacity settings. 

% We construct a semi-synthetic transaction dataset from Rabobank’s public cumulative transaction data~\citep{rabobankdataset}. The original multi-year dataset contains anonymized C2C, C2B, and B2B transfers; we retain these proportions and uniformly redistribute the transactions over a 30-day window. To simulate UPI usage, we assign each user preinstalled apps according to NPCI-reported monthly transaction shares, so that users adopt high-volume apps (e.g., PhonePe, GPay) with higher probability while still reflecting the presence of smaller apps.

\subsection{Baselines}

To benchmark the proposed \textbf{\DTAS}, \textsc{Online\_DTAS}, and \textsc{Fair\_\DTAS} algorithms, we compare them against the following baseline strategies:

\begin{itemize}

\item \textbf{\textsc{DTAS\_No\_Delay}}.
An ablation of \textsc{Online-\DTAS} where the stalling mechanism is removed. Transactions are allocated immediately using the same three-tier strategy, i.e. \textit{preinstalled applications, shared-pool applications,} and \textit{new applications} with affinity-based tie-breaking. This baseline isolates the contribution of delayed scheduling and adaptive heavy-user handling.

\item \textbf{\textsc{Random}}.
A simple allocation strategy where each transaction is assigned uniformly at random among feasible applications. Preinstalled applications are preferred whenever available; otherwise, any application with remaining capacity is selected randomly. This serves as a lower-bound baseline with no optimization or workload awareness.

\item \textbf{\textsc{Least Used}}.
A greedy load-balancing baseline that assigns each transaction to the feasible application with the smallest current load (number of handled transactions), while following the same tiered allocation structure. This baseline explicitly prioritizes load balancing without considering installation overhead.

% \item \textbf{\textsc{ILP Optimal}}.
% For smaller instances, we additionally compare against the optimal ILP formulation introduced earlier. This provides a reference point for evaluating the quality of heuristic solutions in terms of installation efficiency and fairness.

\end{itemize}

\shortversion{
\begin{figure*}[h!]
    \centering
    \begin{subfigure}[b]{0.32\textwidth}
        \centering
        \includegraphics[width=\linewidth]{Pictures/PBAA/transaction_distribution_acm.pdf}
        \caption{Transaction Skewness across users.}
        \label{fig:transaction-skewness}
    \end{subfigure}
    \hfill
    \begin{subfigure}[b]{0.32\textwidth}
        \centering
        \includegraphics[width=\linewidth]{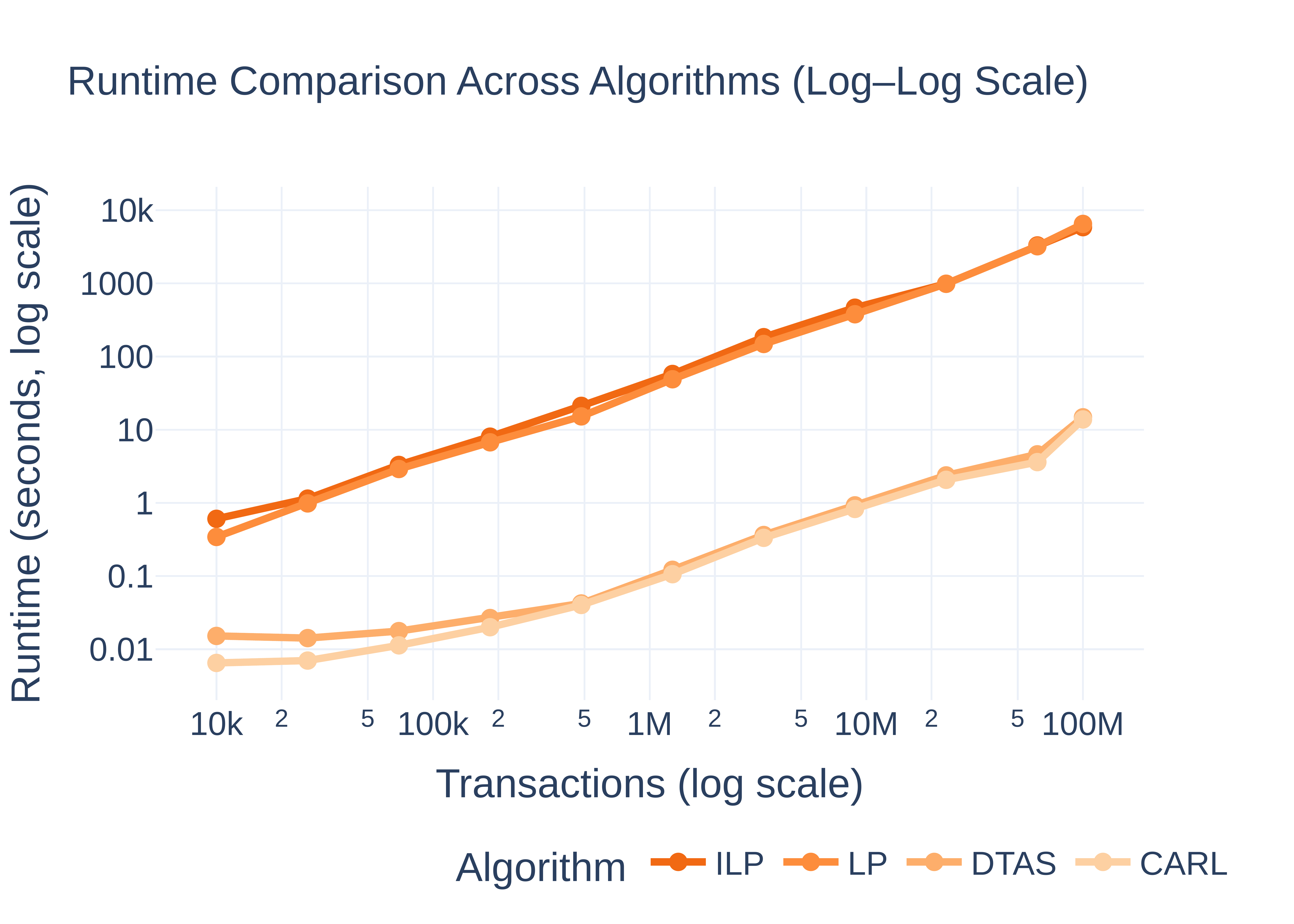}
        % \caption{Execution time comparison of ILP, LP, \CARL, and \DTAS.}
        \caption{Execution Times (ILP, LP, \CARL, and \DTAS).}
        \label{fig:execution-time}
    \end{subfigure}
    \hfill
    \begin{subfigure}[b]{0.32\textwidth}
        \centering
        \includegraphics[width=\linewidth]{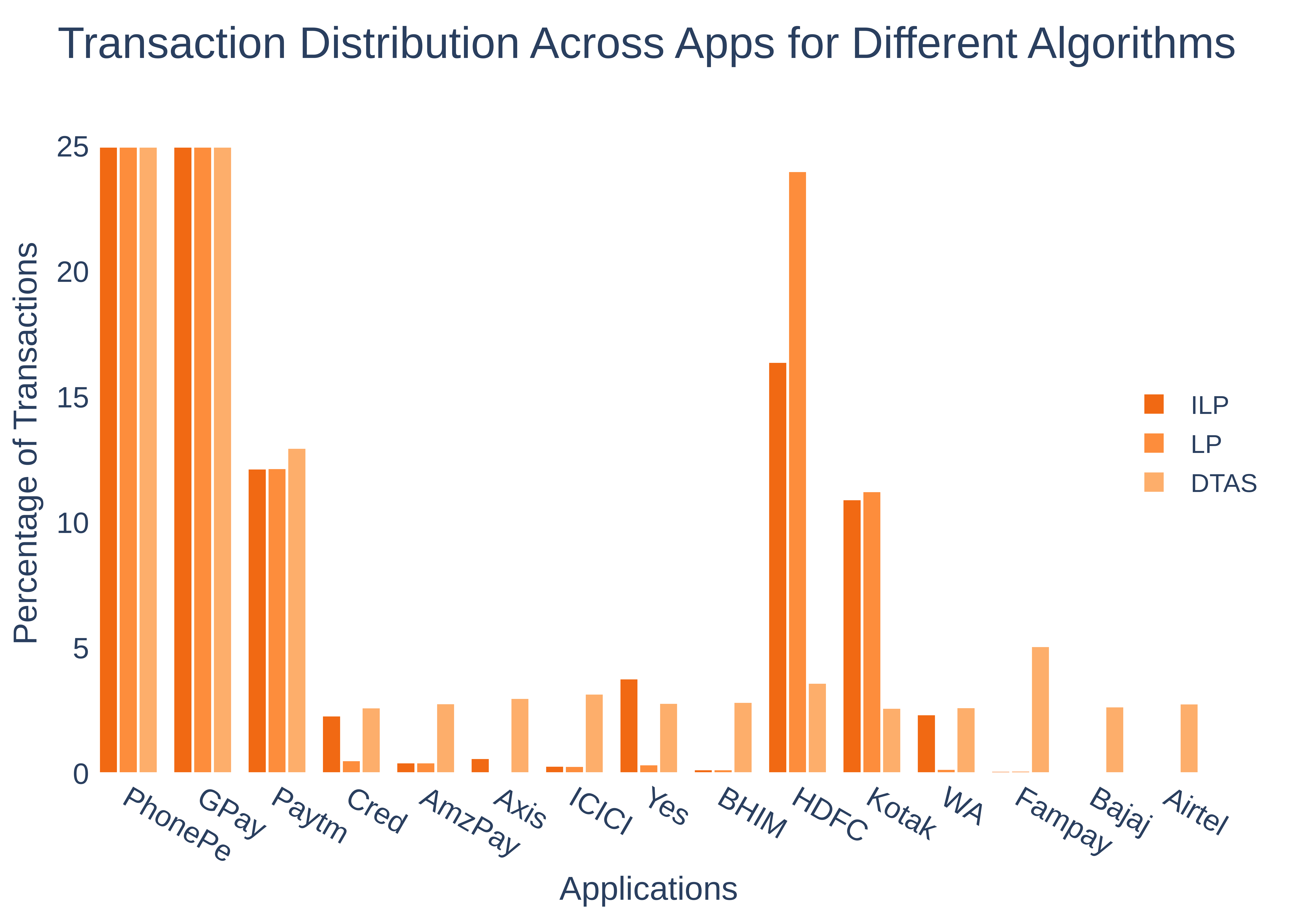}
        \caption{Transaction Allocation across UPI apps.}
        \label{fig:transaction-allocation}
    \end{subfigure}
    
    \caption{Comparison of allocation patterns and execution times across different approaches.}
    \label{fig:combined-figure}
\end{figure*}}
\longversion{\subsection{Evaluation Metrics}
We evaluate all algorithms using the following metrics:

\begin{itemize}

\item \textbf{Per-App Load Distribution.}
Let $\ell_a$ denote the number of transactions routed to application $a$. Since the primary fairness objective is to avoid concentrating transactions on a small subset of applications, we analyze the resulting load distribution across apps. More balanced distributions indicate better utilization of system capacity.

\item \textbf{Extra App Installations.}
We measure the total number of activated edges $|E'|$, corresponding to additional application installations beyond users' preinstalled apps. Lower values indicate greater installation efficiency and reduced user friction.

\item \textbf{Fairness Comparison (Jain's Index).}
To quantify load balancing across applications, we compute Jain's fairness index over application loads:
\[
J(L)=
\frac{\left(\sum_{a\in A}\ell_a\right)^2}
{|A|\sum_{a\in A}\ell_a^2}.
\]

Jain's index ranges from $0$ to $1$, where values closer to $1$ indicate a more uniform distribution of transaction loads across applications.

\end{itemize}}

\shortversion{
We evaluate each algorithm using the following metrics. 
\begin{itemize}
    \item \textbf{Number of App Installations:} Measures the total number of applications required to satisfy all transactions. Lower values indicate higher efficiency.
    \item \textbf{Execution Time:} Represents the total computational time taken to reach a feasible allocation.
    \item \textbf{Fairness Index:} Quantifies how evenly transactions are distributed across applications, measured using the \textbf{inverse Gini score}.
\end{itemize}
}

\begin{table*}[t]
\centering

\begin{subtable}[t]{0.48\textwidth}
\centering
\resizebox{\linewidth}{!}{
\begin{tabular}{lccc}
\toprule
\textbf{Algorithm} & \textbf{Extra Inst.} & \textbf{Jain's Index} & \textbf{Time (s)} \\
\midrule
\textsc{Online\_\DTAS}    & 152,344 & 0.4190 & 17.04 \\
\textsc{Fair\_\DTAS}      & 183,907 & 0.8916 & 130.30 \\
\textsc{\DTAS\_No\_Delay} & 155,671 & 0.4191 & 8.71 \\
\textsc{Random}           & 362,267 & 0.4178 & 7.16 \\
\textsc{Least\_Used}      & 373,117 & 0.4194 & 7.71 \\
\bottomrule
\end{tabular}
}
\caption{Performance comparison of allocation strategies}
\label{tab:freq-results}
\end{subtable}
\hfill
\begin{subtable}[t]{0.48\textwidth}
\centering
\resizebox{\linewidth}{!}{
\begin{tabular}{ccc}
\toprule
\textbf{App Cap.} &
\textbf{Users w/ Remaining Demand (\%)} &
\textbf{Unallocated Transactions} \\
\midrule
10 & 2.25\% & 65,982,913 \\
15 & 1.00\% & 50,982,913 \\
20 & 0.54\% & 40,418,301 \\
25 & 0.29\% & 30,418,301 \\
30 & 0.13\% & 20,418,301 \\
35 & 0.02\% & 10,418,301 \\
40 & $\approx 0\%$ & 1,525,685 \\
\bottomrule
\end{tabular}
}
\caption{Impact of app capacity on remaining demand}
\label{tab:capacity_effect}
\end{subtable}

\caption{Experimental results across allocation performance and capacity sensitivity analyses.}
\label{tab:combined_results}
\end{table*}

% To assess the effectiveness and scalability of the proposed heuristics, we conducted a series of experiments comparing \textbf{ILP}, \textbf{LP Relaxation}, \textbf{\CARL}, and \textbf{\DTAS}. 
% All algorithms were evaluated on synthetic configurations of increasing transaction volumes, ranging from 10,000 to 100 million transactions. 
% Each configuration was tested under identical conditions to ensure fairness in comparison. 

% Since UPI transaction data is not publicly available, we repurpose the Rabobank financial transaction dataset~\cite{saxena2021banking} as a proxy. The original multi-year dataset contains anonymized C2C, C2B, and B2B transfers; we preserve these proportions and uniformly redistribute the transactions over a 30-day window to obtain a realistic daily workload distribution. 
% To simulate UPI behavior, we generate a population of users mapped to Rabobank account IDs and assign each user a set of preinstalled apps based on NPCI-reported monthly transaction shares. We assume there are 15 UPI apps in total. This probabilistic assignment ensures that high-volume apps (e.g., PhonePe, GPay) appear more frequently while still preserving a long tail of smaller apps, giving us a dataset that mirrors real-world UPI adoption patterns.

% \itodobyashlesha{@Shaswat: Please add a short intro to this subsection. It should not start with Figures ..}
\subsection{Quantitative Results}

In this section, we analyze the proposed allocation strategies from the perspectives of fairness, installation overhead, and load distribution across applications. We examine trade-offs between balancing application utilization and minimizing additional installations, while studying how algorithmic choices influence these objectives.

Figures~\ref{fig:sub1} and~\ref{fig:sub2} present Lorenz curves and per-application load distributions across allocation strategies. The Lorenz curves provide a visual interpretation of fairness, where curves closer to the equality line indicate more balanced transaction distributions.

As expected, \textsc{Fair\_\DTAS} remains closest to the equality line, indicating a more uniform allocation than other methods. In contrast, \textsc{Online\_\DTAS}, \textsc{\DTAS\_No\_Delay}, and baseline approaches exhibit greater concentration, consistent with lower fairness.

Table~\ref{tab:freq-results} reports experiments on a workload of 1.5 million transactions and highlights a clear trade-off between fairness and installation overhead. \textsc{Fair\_\DTAS} achieves the strongest fairness among all methods, but this comes with increased installation cost. Conversely, \textsc{Online\_\DTAS} prioritizes minimizing additional installations, resulting in lower fairness but reduced overhead.

The lower fairness observed in \textsc{Online\_\DTAS} is partly driven by skew in application availability across users. Since PhonePe and GPay are preinstalled for many users, and \textsc{Online\_\DTAS} prioritizes already-installed applications before recommending new installations, these applications naturally receive a larger transaction share and higher capacity utilization. Thus, part of the observed concentration arises from user-installation distributions rather than purely algorithmic decisions.

The effect of delaying heavy users can be observed by comparing \textsc{Online\_\DTAS} and \textsc{\DTAS\_No\_Delay}. While both methods exhibit nearly identical fairness behavior, removing the delay mechanism increases installation overhead without meaningful fairness gains. This suggests that selectively postponing heavy users improves allocation efficiency by reducing unnecessary installations.

Baseline methods such as \textsc{Random} and \textsc{Least\_Used} incur higher installation overhead while exhibiting fairness characteristics similar to \textsc{Online\_\DTAS}. In particular, \textsc{Online\_\DTAS} consistently outperforms \textsc{Least\_Used}, indicating that allocation decisions cannot rely solely on selecting applications with maximum available capacity. Instead, applications already installed by a user play a critical role in efficient allocations. This highlights the importance of reuse-aware routing, where prioritizing existing installations significantly reduces overhead without sacrificing allocation quality.

Overall, the results demonstrate a trade-off between fairness and installation efficiency while highlighting adaptive scheduling mechanisms.

\noindent
\longversion{For the experiments in Table~\ref{tab:capacity_effect} and Table~\ref{tab:installations}, we consider the complete dataset consisting of $1.21$ million users and $100$ million transactions to evaluate varying application capacity limits.

Table~\ref{tab:installations} compares additional application installations across transaction scales. \DTAS\ closely matches ILP performance while remaining scalable. Online variants incur higher installation overhead due to lack of future information. \textsc{Fair\_}\DTAS\ improves fairness at increased cost, while increased installations in \textsc{\DTAS\_No\_Delay} highlight the benefit of delaying heavy-user requests.

As shown in Table~\ref{tab:capacity_effect}, increasing app capacity sharply reduces users requiring new installations from approximately $2.25\%$ at a $10\%$ capacity limit to nearly zero at $40\%$. Most users can therefore be fully served under moderate settings, while less than $2.3\%$ experience unmet demand.

However, increased capacity also accentuates duopoly tendencies, where a few dominant applications absorb a disproportionate share of total transaction load. Such concentration may negatively affect competition and resilience. Conversely, lowering app capacity reduces dominance but results in over 65 million unallocated transactions at a $10\%$ limit, indicating potential service degradation under strict caps. If regulatory constraints such as NPCI thresholds were enforced, remaining transactions would directly represent unserved demand, potentially causing user friction and reduced system efficiency.
}

\begin{table*}
\centering
\scriptsize
\resizebox{\textwidth}{!}{
\begin{tabular}{rrrrrrr}
\hline
\textbf{Transactions} &
\textbf{ILP} &
\textbf{LP} &
\textbf{\DTASh} &
\textbf{\textsc{Online\_}\DTASh} &
\textbf{\textsc{Fair\_}\DTASh} &
\textbf{\textsc{\DTASh\_No\_Delay}} \\
\hline

10,000      & 5      & 3.78    & 5        & 26        & 104      & 37 \\

%25,118      & 21     & 20.15   & 21       & 93        & 243      & 182 \\

63,095      & 25     & 24.22   & 25       & 178       & 698      & 533 \\

%158,489     & 20     & 19.31   & 20       & 240       & 1,354    & 903 \\

398,107     & 90     & 88.99   & 93       & 690       & 3,979    & 2,782 \\

1,000,000   & 198    & 196.88  & 198      & 2,031     & 7,771    & 7,799 \\

%2,511,886   & 487    & 485.63  & 487      & 12,101    & 21,194   & 20,052 \\

6,309,573   & 1,314  & 1311.82 & 1,314    & 40,190    & 59,145   & 53,005 \\

%15,848,931  & 3,284  & 3279.25 & 3,284    & 118,971   & 132,170  & 140,028 \\

39,810,717  & 7,676  & 7668.81 & 7,676    & 320,510   & 406,657  & 353,934 \\

100,000,000 & ***    & ***     & 6,599    & 738,726   & 904,179  & 797,000 \\

\hline
\end{tabular}
}

\caption{Comparison of Installations Across Different Algorithms}
\label{tab:installations}
\vspace{1mm}
\begin{minipage}{\textwidth}
\footnotesize
\textit{Note:} LP values correspond to lower bounds obtained from the linear relaxation of the optimization problem. `***` indicates that ILP and LP could not be solved within practical runtime limits. \DTASh\ assumes complete future knowledge and serves as an offline baseline. \textsc{Online\_}\DTASh\ performs transaction allocation without future information. \textsc{Fair\_}\DTASh\ incorporates fairness objectives during allocation. \textsc{\DTASh\_No\_Delay} removes the stalling mechanism and allocates transactions immediately.
\end{minipage}
\end{table*}
\subsection{Practicalities and Discussion}

Given increasing concerns regarding market concentration and the emerging UPI duopoly, transaction-capacity regulations may eventually be enforced by NPCI. Our work provides an initial framework toward implementing such regulations while minimizing disruption to user behavior and market dynamics. Beyond demonstrating feasibility, our results provide a quantified path toward compliance. Across transaction scales, our ILP and DTAS formulations identify the minimum number of incremental user--application connections required to satisfy capacity constraints. At 100 million transactions, \DTAS\ achieves compliance using only 6,599 additional installations while remaining computationally tractable, providing policymakers with practical estimates for sizing incentive budgets such as cashback programs or merchant rewards.

Our sensitivity analysis shows that selecting an appropriate capacity threshold is critical. At a 10\% per-app cap, more than 65 million transactions remain unallocatable and approximately 2.25\% of users require additional installations, whereas both quantities become nearly negligible at 40\%. The proposed 30\% threshold lies in a steep feasibility region where small relaxations substantially reduce unmet demand without significantly altering market structure. This provides NPCI with a principled basis for calibrating thresholds rather than treating a specific percentage as arbitrary.

A practical deployment strategy could begin with a trial phase in which users are not forced to switch applications but instead receive recommendations or incentives encouraging diversification of installed applications. Data collected during this phase can provide insights into user transaction frequencies and preferences required by our framework. Our findings also suggest that intervention strategies should be carefully designed. A counterintuitive but robust observation is that redirecting lightweight users before heavy users substantially reduces installation overhead. Heavy users rapidly consume available capacity, often forcing lightweight users to install new applications despite remaining system-wide capacity. Redirecting lightweight users therefore preserves flexibility and lowers overall system cost.

Interestingly, increasing app capacity improves allocation feasibility while simultaneously worsening concentration. With larger capacity limits, dominant applications absorb a disproportionately larger share of total transaction load, potentially strengthening the same duopoly regulations seek to address. Thus, capacity limits should not merely serve as hard ceilings but as routing targets directing traffic toward underutilized applications.

For regulators prioritizing ecosystem health, \textsc{Fair\_DTAS} provides a useful blueprint. Although it incurs approximately 21\% more installations, it achieves substantially improved fairness and demonstrates that modest deployment costs can produce a more balanced application ecosystem.

Although access to real UPI transaction data remains restricted, our evaluation using a realistic financial transaction dataset suggests that the proposed framework scales effectively and can serve as a practical decision-support mechanism. The framework is also relevant beyond India, particularly as UPI expands internationally and interoperates with payment ecosystems experiencing similar concentration dynamics.

\shortversion{
\noindent
\textbf{Quantitative Results:} As shown in Table~\ref{tab:capacity_effect}, increasing the application capacity sharply reduces the share of users requiring additional installations from 2.25\% at 10\% capacity to just 0.13\% at 30\%. Even modest capacity levels therefore satisfy most user demand, with fewer than 2.3\% of users ever remaining unsatisfied. However, higher capacities also strengthen duopoly tendencies, as a few dominant apps end up handling most transactions. Lowering capacity reduces this concentration but increases the volume of dropped transactions, as reflected in Table~\ref{tab:capacity_effect}. If NPCI enforced strict per-app thresholds, the unmet loads, such as over 65 million dropped transactions at 10\% capacity, would be discarded or delayed, degrading user experience and network reliability.

\begin{table}[h!]
\centering
\caption{Effect of Application Capacity on Remaining Demand and Transactions (Experiment conducted on 1.21M users and 100M transactions)}
\label{tab:capacity_effect}
\begin{tabular}{c c c}
\hline
\textbf{App Capacity} & \textbf{Users Unsatisfied (\%)} & \textbf{Transactions Dropped} \\
\hline
0.10 & 2.25\% & 65,982,913 \\
0.15 & 1.00\% & 50,982,913 \\
0.20 & 0.54\% & 40,418,301 \\
0.25 & 0.29\% & 30,418,301 \\
0.30 & 0.13\% & 20,418,301 \\
0.35 & 0.02\% & 10,418,301 \\
\hline
\end{tabular}
\end{table}

}

\begin{table*}[!t]
\centering
\scriptsize
\resizebox{\textwidth}{!}{
\begin{tabular}{rrrrrrr}
\hline
\textbf{Transactions} &
\textbf{ILP} &
\textbf{LP} &
\textbf{\DTASh} &
\textbf{\textsc{Online\_}\DTASh} &
\textbf{\textsc{Fair\_}\DTASh} &
\textbf{\textsc{\DTASh\_No\_Delay}} \\
\hline

10,000      & 5      & 3.78    & 5        & 26        & 104      & 37 \\

%25,118      & 21     & 20.15   & 21       & 93        & 243      & 182 \\

63,095      & 25     & 24.22   & 25       & 178       & 698      & 533 \\

%158,489     & 20     & 19.31   & 20       & 240       & 1,354    & 903 \\

398,107     & 90     & 88.99   & 93       & 690       & 3,979    & 2,782 \\

1,000,000   & 198    & 196.88  & 198      & 2,031     & 7,771    & 7,799 \\

%2,511,886   & 487    & 485.63  & 487      & 12,101    & 21,194   & 20,052 \\

6,309,573   & 1,314  & 1311.82 & 1,314    & 40,190    & 59,145   & 53,005 \\

%15,848,931  & 3,284  & 3279.25 & 3,284    & 118,971   & 132,170  & 140,028 \\

39,810,717  & 7,676  & 7668.81 & 7,676    & 320,510   & 406,657  & 353,934 \\

100,000,000 & ***    & ***     & 6,599    & 738,726   & 904,179  & 797,000 \\

\hline
\end{tabular}
}

\caption{Comparison of Installations Across Different Algorithms}
\label{tab:installations}
\vspace{1mm}
\begin{minipage}{\textwidth}
\footnotesize
\textit{Note:} LP values correspond to lower bounds obtained from the linear relaxation of the optimization problem. `***` indicates that ILP and LP could not be solved within practical runtime limits. \DTASh\ assumes complete future knowledge and serves as an offline baseline. \textsc{Online\_}\DTASh\ performs transaction allocation without future information. \textsc{Fair\_}\DTASh\ incorporates fairness objectives during allocation. \textsc{\DTASh\_No\_Delay} removes the stalling mechanism and allocates transactions immediately.
\end{minipage}
\end{table*}

\section{Conclusion}

The current UPI ecosystem is heavily dominated by a small number of applications, raising concerns regarding market concentration and motivating interventions such as NPCI's proposed 30\% market-share cap. Enforcing such constraints at scale is computationally challenging, particularly when user installation preferences and application capacities must be respected. In this work, we take a first step toward addressing this problem and propose \DTAS, a scalable allocation framework that preserves users' existing application choices while minimizing additional installations required for compliance.

Our empirical evaluation reveals several important insights. First, reuse-aware allocation substantially reduces installation overhead compared to naive load-spreading strategies. Second, application availability and existing user installations play a critical role in allocation quality, highlighting the importance of incorporating user-level constraints into routing decisions. Third, heavy-user behavior strongly influences allocation efficiency: prioritizing or immediately serving heavy users rapidly saturates popular applications and reduces future reuse opportunities. This observation motivates the delay mechanism in \textsc{Online\_DTAS}, where postponing dynamically identified heavy users improves long-term allocation efficiency without significantly affecting fairness. Finally, fairness and installation efficiency exhibit an inherent trade-off, where stronger fairness objectives improve load balancing at the cost of additional installations.

Overall, \DTAS\ and its online variants achieve performance close to optimization-based approaches while remaining computationally efficient at large scales, even when optimization becomes impractical. These findings suggest that scalable reuse-aware allocation strategies can serve as practical decision-support tools for regulators seeking to improve ecosystem balance while preserving user convenience.

Our current formulation focuses on balancing transaction counts across applications. An important future direction is incorporating transaction values into the allocation process. High-value transactions may introduce different capacity requirements, user preferences, and financial risks, potentially revealing richer trade-offs between fairness, efficiency, and ecosystem resilience.

\section*{Acknowledgements}
We thank Sreemanti Dey for her valuable help with an earlier version of this work. We are also grateful to Akrati Saxena for sharing the Rabobank transaction dataset, which significantly facilitated our empirical evaluation.

\bibliographystyle{ACM-Reference-Format}
\bibliography{main}

\appendix

\end{document}